\documentclass[12pt,fleqn,reqno]{amsart} %[12pt,reqno]{amsart} %[12pt,twoside]

\usepackage{epsfig}
\usepackage{graphicx}
\usepackage{enumerate}
\usepackage{color} 
\usepackage[toc,page]{appendix}

\usepackage[dvipsnames]{xcolor}
\usepackage{bm,bbm}
\usepackage{cancel}
\usepackage{enumitem}

\usepackage{amsmath,euscript}
\usepackage{amsthm}
\usepackage{amssymb}
\usepackage{graphicx}
\usepackage{mathrsfs}
\usepackage{epsf}
\usepackage{xcolor}
\usepackage{psfrag}
\usepackage{enumerate}
\usepackage{amsfonts,amssymb,amsthm,amsmath} %,latexsym} %,marginnote}
\usepackage{hyperref}

\usepackage[numbers]{natbib}
\usepackage{enumitem} % for labeling in enumeration
\usepackage{xspace}
\usepackage[margin=1.37in]{geometry}
\usepackage{changes}
\definechangesauthor[name={SB}, color=magenta]{sb}
\definechangesauthor[name={JF}, color=blue]{jf}
\definechangesauthor[name={ML}, color=red]{ml}
\definechangesauthor[name={MS}, color=green]{ms}
\usepackage{soul}

\newtheorem{thm}{Theorem}

\newtheorem{lemma}[thm]{Lemma}
\newtheorem{prop}[thm]{Proposition}

\theoremstyle{remark}
\newtheorem{remark}[thm]{Remark}

\theoremstyle{definition}

\numberwithin{thm}{section}
\numberwithin{equation}{section}
\usepackage{soul}
\setstcolor{red}

\definecolor{green}{rgb}{0.0, 0.5, 0.5}
\definecolor{yellow}{rgb}{0.5, 0.5, 0}
\definecolor{lgray}{gray}{0.9}
\definecolor{llgray}{gray}{0.95}
\definecolor{lllgray}{gray}{0.975}

\newcommand{\nc}{\newcommand}

\nc{\la}{\label}
\nc{\ba}{\begin{array}}
\nc{\ea}{\end{array}}
\nc{\bs}{\begin{split}}
\nc{\es}{\end{split}}

\newcommand{\R}{\mathbb{R}}
\newcommand{\C}{\mathbb{C}}
\newcommand{\Z}{\mathbb{Z}}

\newcommand{\cF}{\mathcal{F}}
\nc{\al}{\alpha}
\nc{\del}{\delta}
\nc{\h}{\delta}
\nc{\G}{\Gamma}
\nc{\et}{\eta} 
\nc{\g}{\gamma}
\nc{\gam}{\gamma}
\nc{\ka}{\kappa}
\nc{\lam}{\lambda}
\nc{\Lam}{\Lambda}
\nc{\Om}{\Omega}
\nc{\om}{\omega}

\nc{\ta}{\tau}
\nc{\w}{\omega}
\nc{\io}{\iota}
\nc{\z}{\zeta}
\nc{\s}{\sigma}
\nc{\Si}{\Sigma}
\nc{\vphi}{\varphi}
\nc{\bP}{\bar{P}}
\nc{\bQ}{\bar{Q}}
\nc{\ran}{\rangle}
\nc{\lan}{\langle}

\newcommand{\ls}{\lesssim}

\newcommand{\Ran}{\operatorname{Ran}}

\newcommand{\supp}{\operatorname{supp}}

\renewcommand{\Re}{\mathrm{Re}} % Real part
\renewcommand{\Im}{\mathrm{Im}} % Imaginary part

\newcommand{\im}{{\rm Im}}
\newcommand{\Tr}{\mathrm{Tr}}
\newcommand{\tr}{\mathrm{Tr}}

\nc{\bfone}{{\bf 1}}
\nc{\1}{{\bf 1}}

\newcommand{\p}{\partial}

\newcommand{\n}{\nabla}

\newcommand{\DETAILS}[1]{}

\newcommand{\x}{\lan x\ran}

\nc{\den}{\text{den}}
\nc{\ex}{\text{xc}}
\nc{\Ex}{\text{Xc}}

\title[Maximal Propagation Speed in Open Quantum Systems]{Maximal Speed of Propagation in Open Quantum Systems}

\date{February 4, 2022}

 \author[S.~Breteaux]{S\'ebastien Breteaux}
\address{Institut Elie Cartan de Lorraine, Universit\'e de Lorraine, 57045 Metz Cedex 1, France}
 \email{sebastien.breteaux@univ-lorraine.fr}

 \author[J.~Faupin]{J\'er\'emy Faupin}
\address{Institut Elie Cartan de Lorraine, Universit\'e de Lorraine, 57045 Metz Cedex 1, France}
 \email{jeremy.faupin@univ-lorraine.fr}
 
 \author[M.~Lemm]{Marius Lemm}
 \address{ Department of Mathematics, University of T\"ubingen, 72076 T\"ubingen, Germany}
 \email{marius.lemm@uni-tuebingen.de}

 \author[I.~M.~Sigal]{Israel Michael Sigal} %\footnote{B}
 \thanks{The research of IMS is supported in part by NSERC Grant No. NA7901.}
 \address{Department of Mathematics, University of Toronto, %40 St.~George Street, 
Toronto, M5S 2E4, %Ontario, 
Canada}
 \email{im.sigal@utoronto.ca}

\begin{document}

 %for  von Neumann-Lindblad equation

%\begin{dedication} \qquad {To Elliott Lieb in recognition ... %with friendship and admiration }\end{dedication}
 
\begin{center}

\begin{quote} \qquad \qquad \qquad  	{\it To Elliott Lieb in recognition of his groundbreaking\\  \, \qquad \qquad \qquad work in mathematical and theoretical physics % with admiration and friendship.
 } \end{quote}
 \end{center}

\bigskip

\maketitle

\begin{abstract}
We prove a maximal velocity bound for the dynamics of Markovian open quantum systems. The dynamics are described by  one-parameter semigroups of quantum channels %generated by a Lindbladian. 
satisfying the von Neumann-Lindblad equation. Our result says that dynamically evolving states are contained inside a suitable light cone up to polynomial errors. We also give a bound on the slope of the light cone, i.e., the maximal propagation speed. The result implies %can be interpreted as
 an upper bound on the speed of propagation of local perturbations of stationary states %{\bf [information/state]} %quantum correlations
  in open quantum systems.

%The dynamics of MOQS open quantum systems are given as a one-parameter  and so this result shows that the quantum information is transmitted with a finite speed and gives an explicit  bound on the maximal speed  of the transmission.

%This work is inspired by  Lieb-Robinson bounds on (effective) propagation of  quantum correlations which provided a powerful tool in many areas of quantum physics. 

 \bigskip

%\subjclass
\noindent Subjclass: {Primary 35Q40; Secondary 35P25; 81P47; 81U99}

%\keywords
  \end{abstract}

%\footnote{The work on this paper was supported in part by NSERC Grant No. NA7901 (IMS).} 

   % \newpage
 \section{Introduction} 

 In this paper, we prove a maximal velocity bound for the dynamics of Markovian open quantum systems. 
 
 Our work is inspired by the celebrated Lieb-Robinson bounds  \cite{LR}  on propagation of quantum correlations in quantum spin systems. The Lieb-Robinson bounds established a fundamental physical principle in Statistical Mechanics by showing rigorously that quantum correlations (specifically, commutators of local observables) are restricted to an effective light cone in space-time.  They also provided an effective tool in many areas of quantum physics. 
 
 For examples of breakthroughs in quantum many-body theory that leveraged Lieb-Robinson bounds in essential ways, we mention (i) Hastings' proof of the area law for the entanglement entropy for ground states of 1D gapped Hamiltonians \cite{H07}, (ii) the proofs by Hastings-Koma and Nachtergaele-Sims of the folklore assertion that a gap leads to exponential decay of correlations \cite{HK,NS1}, and (iii) the modern classification of topological quantum phases \cite{BMNS,H04,HW}. 
 
 Many other applications of Lieb-Robinson bounds have been found since then in areas as diverse as condensed-matter physics, quantum information science and high-energy physics to name a few \cite{BdRF,BHM,BHV,KS_he,LVV,NS_ls,RS}. Effective light cones of Lieb-Robinson type have been observed in the laboratory \cite{exp1,exp2} and in sophisticated numerical experiments \cite{nexp1,nexp2}, typically starting from a quench of the system. 

The supreme usefulness of Lieb-Robinson bounds has also led to many extensions and variants of the original result. Indeed, exploring the scope of Lieb-Robinson bounds has developed into its own branch of research; see, e.g., \cite{Aetal,DLLY1,DLLY2,DLY,DV,EKS,FLS1,FLS2,F,GL,GMN,GNRS,H22,HSS,K,KS,NOS,NRSS,NSY1,NSY2,Pou,SHOE,WH,YL}. Recent efforts in this highly active  area have been especially focused on extensions to fermionic \cite{GNRS,NSY2} or bosonic \cite{FLS1,FLS2,KS,SHOE,WH,YL} Hamiltonians and to long-range interactions \cite{EMNY,Fossetal,KS_he,MKN,Tranetal}. See also \cite{GMN,H22} for very novel directions. We particularly emphasize the works \cite{DV, NVZ} and %the paper by Poulin 
 \cite{Pou} since they concern open quantum systems, which are also the topic of this paper. For more information on Lieb-Robinson bounds, we recommend the reviews \cite{H10,KGE,NS2}. %To summarize, it is not an exaggeration to count Lieb-Robinson bounds among the most consequential contributions of mathematical physics to physics at large.

  In an independent later development in $n$-body quantum mechanics,   it has been demonstrated in \cite{SigSof} that, up to small probability tails vanishing with time, the supports of wave function solutions of the Schr\"odinger equation spread with a finite speed.    
  This result was further improved in \cite{APSS,HeSk,Skib}, with \cite{APSS} proving  an energy dependent bound on the maximal speed of propagation. It was  extended in \cite{BonyFaupSig} to photons interacting with an atomic or molecular system (see also \cite{DerGer2,  FrGrSchl2, FrGrSchl3, FrGrSchl4, Ger}) %\jf{(I removed \cite{DeRoGriKu} from this list and added \cite{Ger})}. %, DeRoKu     
 The above bounds were used in a fundamental way in the scattering theory (see \cite{DeRoGriKu,Der, DerGer, DerGer2, FaSig, FrGrSchl2, FrGrSchl3, FrGrSchl4, Ger, HuangSof, SigSof2}). 
 Furthermore, \cite{FLS1, FLS2} and \cite{AFPS} developed related techniques in condensed matter physics  to prove the maximum velocity bounds for transport of particles in the Bose-Hubbard model in the thermodynamic regime and in the nonlinear Hartree many-body mean-field dynamics, respectively. 
In this paper, we extend this approach to Markov open quantum systems.

 The link between this approach and Lieb-Robinson bounds was made in \cite{FLS2} when a Lieb-Robinson bound was proved for the Bose-Hubbard model by similar techniques.
For more on the relation of the velocity bounds that we obtain here to Lieb-Robinson bounds %will be explained \ml{in more detail} at 
see the end of Subsection \ref{ref:result}.

 \subsection{The von Neumann-Lindblad equation}
 The dynamics of open quantum systems originate from   the unitary dynamics of systems interacting with an environment by tracing out the latter. States of such systems 
  are described by  density operators $\rho$, i.e. positive trace class operators, $\rho=\rho^{*}\geq 0,\hspace{0.2cm}\Tr(\rho)<\infty$, on some Hilbert space $\mathcal{H}$.

  We are interested in open quantum dynamics under the usual Markov (semi-group) assumption.
It has been proven in  \cite{GKS,MR0413878} that, for finite-dimensional Hilbert spaces,  Markovian open quantum dynamics satisfy the   von Neumann-Lindblad (vNL) equation %(see \cite{MR0413878})   
 \footnote{Here and in what follows we use the units in which the Planck constant is set to $2\pi$, and the speed of light to 1: $\hbar=1$ and $c=1$.}  %\footnote{{\bf This conjecture has been proven in  \cite{MR0413878}. For infinite-dimensional Hilbert space, the converse statement, i.e. that  the dynamics generated by \eqref{vNLeq} is a quantum dynamical semigroup, was shown in \cite{kossa,IngardenKossakowski,Davies}, see Remarks \ref{rem:quant-info} and \ref{rem:exist} below.}} 
  \begin{align}\label{vNLeq}
	&\frac{\partial\rho_t}{\partial t}=-i[H,\rho_t]+\frac12\sum_{j\geq 1}\big([W_{j},\rho_t W_{j}^{*}]+[W_{j}\rho_t,W_{j}^{*}]\big). %(W_{j}\rho W_{j}^{*}-\frac{1}{2}\{W_{j}^{*}W_{j},\rho\}),	%&L(\rho):=-i[H,\rho]+G(\rho),
\end{align}
Here   $H$ is a self-adjoint operator on $\mathcal{H}$,  the quantum Hamiltonian of a proper quantum system, and 
%\begin{equation}\label{G}	G(\rho):=\sum_{j}W_{j}\rho W_{j}^{*}-\frac{1}{2}\{W_{j}^{*}W_{j},\rho\},\end{equation} and
  $W_{j}$ are bounded operators. We assume that $\sum\nolimits_{j\geq 1}W_{j}^{*}W_{j}$ converges in the space of bounded operators $\mathcal{B}(\mathcal{H})$.  
  In what follows, we always assume the above properties of $H$ and $W_j$.  
 
 We deal with \eqref{vNLeq} for infinite-dimensional Hilbert spaces, where it is conjectured that  the statement above is true as well. In any case, it was shown in \cite{Davies,IngardenKossakowski,kossa} that  the converse statement, i.e. that  the dynamics generated by \eqref{vNLeq} is a quantum dynamical semigroup, holds for both finite-dimensional and infinite-dimensional Hilbert spaces. Hence, if we wish to avoid the conjecture, we may confine ourselves to considering  Markovian open quantum dynamics generated by the vNL
%  von Neumann-Lindblad (vNL) 
equations.  
 
 For a discussion of existence results and, in particular, for a definition of the weak solution used in the main theorem below, see Subsection \ref{sect:exist}. For a discussion of open quantum systems and irreversibility, see \cite{Fr, GS}.

 \subsection{Main result}\label{ref:result}

Now we suppose that $\mathcal{H}=L^2(\mathbb{R}^d)$. For a Borel set $A$, we let $\chi_A$ denote the characteristic function of $A$. For an operator $B$, the symbol $\mathcal{D}(B)$ denotes the domain of $B$.

We are interested in proving that the solution $\rho_t$ to \eqref{vNLeq} obeys a {\it maximal propagation speed bound (MSB)}. By this we mean that, under an appropriate localization assumption on the initial state, there is a constant $c<\infty$ s.t.~the probability, $\Tr(\chi_{|x|\ge ct}\rho_t)$, that  the system 
 is localized in the domain $\{|x|\ge ct\}$ vanishes, as $t\to\infty$:
\begin{align} \label{max-speed-def}
\Tr(\chi_{|x|\ge ct}\rho_t)\to 0. 
\end{align}
In fact, the main result will allow for a non-localized stationary part and is thus more general.
Moreover, the scalar  $c_{\rm max}:=\inf \{ c:$ \eqref{max-speed-def} holds$\}$ will be called  the {\it maximal propagation speed}. 

 \medskip
 
In this article we make no distinction in our notation between functions and the operators of multiplication defined by those functions.

 \medskip
%\begin{assmpt}\label{assmpt:main}
\paragraph{\bf Assumptions.} Denote $\langle x\rangle = \sqrt{1+|x|^2}$. We assume that
\begin{equation}\label{Hassmpt:main'} %{domainassmpt}
\langle x\rangle^{-1}\mathcal{D}(H)\subset\mathcal{D}(H),
\end{equation}
and that
  for some positive integer $n$, the following estimates on $H$ and $W_j$ hold
\begin{align}
\label{Hassmpt:main}&\|\mathrm{ad}_{\x}^k(H)\|\ls 1, \quad \mathrm{for}\quad  1\leq k \leq n, \\
\label{Wassmpt:main}\sum_{j\geq 1} &\|\mathrm{ad}_{\x}^k W_j\|^2\ls 1, \quad \mathrm{for} \quad   1\leq k \leq n. %\sum_{j\geq 1} \sum_{j\geq 1} \max_{\substack{0\leq \ell,m\leq n:\\ \ell+m=n}} \|ad_{\x}^\ell W_j^*\| \|ad_{\x}^m W_j\|\ls& 1.
\end{align}
%\end{assmpt}

Here the commutators $\mathrm{ad}_{\x}^k(H)$ are defined recursively by $\mathrm{ad}_{\x}^0(H)=H$ and, for all integer $k$, $\mathrm{ad}_{\x}^{k+1}(H)=[\mathrm{ad}_{\x}^k(H),\x]$.\footnote{As usual, the commutator between two operators $A$, $B$ is defined as a quadratic form on $\mathcal{D}(A)\cap\mathcal{D}(B)$. Assumptions \eqref{Hassmpt:main}--\eqref{Wassmpt:main} postulate that the commutators extend to elements of $\mathcal{B}(\mathcal{H})$. }

We describe examples of $H$ and $W_j$ of interest in Subsection \ref{sect:discussion} below.
By  Assumptions \eqref{Hassmpt:main}-\eqref{Wassmpt:main}, with $k=1$, the operator 
\begin{align}
\label{gam}\g:=i[H,\x]+\frac{1}{2}\sum_{j\geq 1} \big(W_j^* [\x, W_j]+[W_j^*, \x] W_j\big)\end{align} 
extends to a bounded operator. Its physical meaning is discussed in Subsection \ref{sect:discussion} below. Its norm
  \begin{align} \label{kappa} 
 \kappa:=\left\| \g\right\|,
\end{align} 
 will give a bound on the maximal propagation speed. 
 We introduce the regions and the corresponding characteristic functions
\begin{equation*}
A_\eta:=\{x\in \R^d: \langle x \rangle\ge \eta\}, \quad \textnormal{ and }\quad\chi_b:=\chi_{A_{b}}.
\end{equation*}

We say that a state $\rho_\mathrm{st}$ is a static solution to  \eqref{vNLeq} if it is a time-independent bounded operator that solves \eqref{vNLeq}.

  Our main result is the following theorem. 
  
\begin{thm}[Maximal propagation speed bound] \label{thm:msb}
Suppose that Assumptions \eqref{Hassmpt:main'}--\eqref{Wassmpt:main} hold for some positive integer $n$. Let  $\rho_0 := \rho_{\rm st} + \lam$, where $\rho_{\rm st}\ge 0$ is a static solution to \eqref{vNLeq} and $ \lam$ is a trace-class operator s.t. $ \lam\ge 0$ (or $-\rho_{\rm st}\le \lam\le 0$) and $\chi_b \lam=0$ for some $b>0$. Then,  for all $a>b$, $c>\kappa$, there exists $C_n>0$ such that  the unique weak solution $\rho_t$ to \eqref{vNLeq}  with the initial condition $\rho_0$ satisfies the estimate
\begin{equation} \label{max-vel-est}
 \Tr(\chi_{{\eta}}\, \rho_t)\leq \, C_n \eta^{1-n} + \Tr(\chi_{{\eta}} \, \rho_\mathrm{st}),\qquad \textnormal{for all }\eta \geq a + ct,\, t>0.
 \end{equation}
\end{thm} 

%\bigskip

In a nutshell,  Theorem \ref{thm:msb} says that under the vNL dynamics, the leakage of the particle probability outside of the light cone $\eta\sim a + ct$ is polynomially suppressed for any $c>\kappa$. In other words, $\kappa$ bounds the maximal propagation speed of particles. We remark that the initial condition $\rho_0=\rho_{\rm st}+ \lam$ 
  appearing in Theorem \ref{thm:msb}  is not localized around the origin, unless $\rho_{\rm st}=0$.
  
To interpret the result, we recall that the dynamics %of open Markov quantum systems, i.e. the one 
generated by the vNL equation,  
are given by  linear, strongly continuous, one-parameter semigroups of trace-preserving and completely positive contractions, called quantum dynamical semi-groups. (A converse statement was proven, for finite-dimensional Hilbert spaces, in \cite{MR0413878}.)   As  linear, completely positive maps define quantum channels with quantum information encoded in density operators, the  vNL equation could be interpreted as describing transmission of quantum information along a quantum channel defined by the vNL equation \eqref{vNLeq}. %generator $L$ defined by the r.h.s. of \eqref{vNLeq}.  
   Then estimate \eqref{max-vel-est} establishes that the quantum information is transmitted with a finite speed and gives an explicit  bound on the maximal speed  of the transmission.

  This result may be compared to the MSB for the Schr\"odinger equation (\cite{APSS}), on the one hand, and the LR bounds with   the Lindblad term (\cite{Pou}), on the other.

To compare our bounds with the Lieb-Robinson ones, %we mention that both deal with the maximal propagation speed of perturbations. In the Lieb-Robinson case, one looks at  the propagation of correlations, in our case, of the localization of solutions/solutions.
the latter deal with the propagation of correlations %, while ours, with the propagation of localization of probabilities.  
in quantum statistical mechanics of macroscopic (or bulk) systems, %while we, with 
 %From a substantial, physical viewpoint,  Lieb and Robinson deal with
while ours deal with the propagation of localization of probabilities in quantum mechanical systems at the zero density, i.e. with a finite number of particles propagating in an infinite physical space.  %Quantum Mechanics, while Lieb and Robinson, with Quantum Statistical Mechanics. %while More specifically, one could think of our system at the zero density one with a finite number of particles propagating in an infinite physical space. % and consequently without control of constants in the dimension interactions are allowed to be  

On a technical level, our approach works in both continuous and discrete cases and for unbounded interactions, while with exception of \cite{GNRS} and \cite{NRSS}, the Lieb-Robinson bounds are obtained for discrete Hamiltonians and bounded interactions. %of quantum statistical mechanics which, as in the case of spin systems, might consist entirely of interactions. 

 Moreover, we allow rather general interactions which could be of $N$-body type. While dependence of  constants on the dimension, i.e. on the number of particles $N$,  is not controlled here, our techniques are adaptable to the  quantum statistical mechanics setting as shown in \cite{FLS1, FLS2}.

% here is a certain parallel between the MPS bound and signals information %(defined by certain thresholding) 
%   in quantum systems over large distances. % and both estimate similar objects.
 %  However, there is also a principal difference between the two.  The MPS bound is determined by the maximal energy available  to the state and requires no special local structure for interactions, %rather arbitrary interactions, % of the initial condition available (and the relative bound on the interaction), 
  % while the Lieb-Robinson one does not depend on the energy and is due entirely to the locality of the interaction. %\footnote{}  

%Lieb-Robinson bounds estimate the group velocity associated to a many-body quantum dynamics. The models you are investigating are defined on a single Hilbert space and that Hilbert space has no a priori local tensor product structure. Beyond the phrase ``propagation estimate'', is there a better way to understand the relation between these two estimates? 

\subsection{Existence of solutions to the von Neumann-Lindblad equation}\label{sect:exist}   
 Denote by $S_{1}$ the Schatten space  of trace class operators. Let $L$ %:\mathcal{D}(L)\to S_1$
  be the operator on $S_{1}$ defined by the r.h.s. of \eqref{vNLeq},  i.e.,
 \begin{equation}\label{L}
 L\rho=-i[H, \rho]+\frac12\sum_{j}\big([W_{j},\rho W_{j}^{*}]+[W_{j}\rho,W_{j}^{*}]\big), %\sum_{j=0}^{\infty}(W_j \rho W_j^*-\frac{1}{2}\{W_j^* W_j,\rho\}),
\end{equation}
with the domain $\mathcal{D}(L)= \mathcal{D}(L_0)$, where $L_0\rho :=-i[H, \rho]$, or explicitly 
 \begin{align}
\mathcal{D}(L):=\, \big\{\rho &\in  S_1\, | \, \rho\mathcal{D}(H)\subset\mathcal{D}(H) \text{ and } H\rho-\rho H \in  S_1 \big \} \subset S_1.
%\notag \\ &H\rho-\rho H \text{ defined on }\mathcal{D}(H) \text{ extends to an element of } S_1 \big \}, \label{eq:def_domain}
 \end{align}
  %For strong solutions, we require that  initial conditions to be from $\mathcal{D}(L)= \mathcal{D}(L_0)$. 
Let $L'$ be the  operator on the space of observables $\mathcal{B}(\mathcal{H})$ dual  of $L$ with respect to the coupling $(A, \rho):= \tr(A \rho)$, i.e. \[\tr(A L\rho)= \tr((L'A) \rho),\] for $\rho\in\mathcal{D}(L)$ and $A\in\mathcal{D}(L')\subset\mathcal{B}(\mathcal{H}) $ 
 (see \eqref{L'} for an explicit expression).\footnote{$L'$ generates  the dual Heisenberg-Lindblad  evolution
$ \partial_t A_t= L' A_t$ of quantum observables.} We say that \eqref{vNLeq} has a weak solution $\rho_{t}$ in $S_1$, if for any observable $A$, i.e.~$A\in\mathcal{B}(\mathcal{H}) $, in the domain of the operator $L'$, we have \[\tr\left(A \frac{\partial\rho_t}{\partial t}\right)= \tr((L'A) \rho)\] (see e.g. \cite{OS}).
By a standard argument,  for any initial condition $\rho_{0}\in S_1$, \eqref{vNLeq} has a unique weak solution in $S_{1}$ 
 (see e.g. \cite[Section 5.5]{Davies}, \cite[Appendix A]{FFFS} or \cite{OS} for a detailed discussion). 
One can show further (see \cite{AlickiLendi,Davies,FFFS,IngardenKossakowski,kossa}) that $L$ defines a completely positive, trace preserving, strongly continuous semigroup of contractions so that, %which preserves the trace,
 in particular:
\[\rho_{t}\ge 0 \quad \text{ if } \quad  \rho_{0}\ge 0,  \quad  \text{ and } \quad\tr\rho_{t}=\tr\rho_{0}.\]

 % \bigskip

\subsection{Discussion of Theorem \ref{thm:msb}} %main result
\label{sect:discussion}

%\begin{enumerate}[label=(\roman*)]
%\item

The operator $\gamma$ given in  \eqref{gam} can be formally written as \[\g=L'\x.\] 
It essentially represents  the component of the velocity operator $L' x$ along $x$. We expect that in many circumstances the environment will produce such quantum ``friction'' 
 and even lead to equilibration (see \cite{BFS,FrMe,JaPi,Rob,Spohn} for analysis for quantum systems of finite degrees of freedom).  Thus, we formulate the following conjecture:
 
 \medskip
 
\paragraph{\textit{Conjecture}.}  
For generic $W_j\neq 0$ and $H$, it holds that $\kappa<\|\mathrm{ad}_{\x} (H)\|$.

 \medskip

This conjecture would imply that the maximal propagation speed is smaller than $\|\mathrm{ad}_{\x} (H)\|$. A weaker version of the conjecture would be that the maximal propagation speed of any open quantum system with $W_j\neq 0$ is bounded by $\|\mathrm{ad}_{\x} (H)\|$. (This weaker version would be implied by the conjecture stated before since, by the result presented here, $\kappa$ bounds the maximal propagation speed of the open quantum system.)

%\subsection{Specific operators $H$}\label{ssect:H-W} 
Let us now discuss specific choices for the operators $H$ and $W_j$.

(A) The key example of the operator $H$ is the Schr\"odinger-type operator \begin{align} \label{H-ex}H=\om(p)+V(x),\end{align} with momentum operator $p:=-i\n$. To satisfy our assumptions, we require that the for the kinetic energy symbol $\om$ that $|\p^\al\om(\xi)|\ls 1$ for $1\leq |\alpha| \leq n$ and for the potential $V(x)$ that it is $\om(p)$-bounded with the relative bound $<1$. We recall that relative boundedness means that
\begin{align} \label{V-cond}
& \exists \,0\le a_1<1,\ a_2>0: \quad
\|Vu\|\le a_1 \| \om(p) u\|+ a_2\|u\|, 
\end{align}
with $\|\cdot\|$ being the norm in $L^2(\R^d)$.
By the Kato-Rellich theorem  (see e.g. \cite{CFKS}), these assumptions ensure that $H$ is self-adjoint on the domain of $\om(p)$.
 
 (B)  Another example in which the operators $H$ satisfy our assumptions arises if we consider the vNL equation on $\Z^d$, where the derivatives are automatically bounded.

 (C) Examples of the Kraus-Lindblad  operators  $W_j$ such that $\mathrm{ad}_{\x}^k W_j$ are bounded are provided by  pseudodifferential operators, $W_j=w_j(x, p)$  with symbols $w_j(x, \xi)$ satisfying the estimates 
\[|\p_x^\beta \p_\xi^\al w_j(x, \xi)|\ls \lan \xi\ran^{- \del' |\al|+\del|\beta|},\]
 for $0\le\del<\del' \le 1 $ and the multi-indices $\al$ and $\beta$, with $|\al|$ and $|\beta|$ sufficiently large.

 (D)  For specific physical examples of the Kraus-Lindblad operators $W_j$, see \cite[Section 4]{FFFS}.

We note that in its current form, the assumption $|\p^\al\om(\xi)|\ls 1$ for $1\leq |\alpha| \leq n$ excludes the Laplacian and thus the standard Schr\"odinger operator. The underlying reason is Assumption \eqref{Hassmpt:main}, with $k=1$, which requires good ultraviolet behaviour. 

\medskip

\paragraph{\textit{Open problem}.} Relax Assumption \eqref{Hassmpt:main}, with $k=1$, to $H$-boundedness of  $\mathrm{ad}_{\x}(H)$. (In this case, we can allow $\om(p)$ in \eqref{H-ex} satisfying $|\p\om(\xi)|\ls \om(\xi)$ instead of $|\p\om(\xi)|\ls 1$ and  therefore the standard Schr\"odinger operators, with  $\om(p)=|p|^2=-\Delta$.)

\begin{remark} The scattering theory for von Neumann-Lindbald equations generated by unbounded operators has been studied in \cite{Alicki1,AlickiFrigerio,Davies2,FFFS,FaFr18_01}. Assuming that $H$ has purely absolutely continuous spectrum and that the operators $W_j$ satisfy a suitable smallness condition, it is proven in these references that the dynamics given by \eqref{vNLeq} asymptotically converge, as $t\to\infty$, to the free, Hamiltonian dynamics given by  the von Neumann equation 
\begin{align}\label{vNLeq_free}&\frac{\partial\rho_t}{\partial t}=-i[H_0,\rho_t]\end{align}
where $H_0$ is the Hamiltonian of the closed system.
Note that under these special conditions, a weak version of the maximal velocity bound can be deduced from this scattering result. In the general case, the scattering theory must be modified to allow for a description of the phenomenon of absorption, or capture \cite{AlickiFrigerio,Davies2,FFFS}.

For earlier results on the problem of return to equilibrium, see \cite{Bardetetal,BaNarTh,BaNar,FaRe,Rob, Spohn}.  %{\bf(more recent references? \jf{I added references \cite{BaNar,BaNarTh,FaRe})}}
\end{remark}

At the core of our proof lies a construction of  {\it propagation observables} satisfying the {\it recursive monotonicity estimate} (RME). Section \ref{sec:msb-pf}  presents the general method we use without using specificities of the dynamics. In this section, we derive Theorem \ref{thm:msb} from the RME. The following  Section \ref{sec:rme-pf} contains the proof of the RME and this is where the specific dynamics we consider enter. %This paper is organized as follows. Theorem \ref{thm:msb} is proven in Section \ref{sec:msb-pf} given the recursive monotonicity estimate (RME). The latter is proven in Section \ref{sec:rme-pf}.
In Appendix \ref{sec:commut}, we present, for convenience of the reader,  known results on operator functional calculus, namely,  expansions  of commutators of operator functions  with estimates of the remainders.  

%We emphasize that  Section \ref{sec:msb-pf} presents the general method we use. It is based on the RME and does not use specificity of the dynamics we consider. The latter enters into the proof of the RME in the following section.

\bigskip

\textit{Notation.} We write $\|\cdot\|$ for the operator norm.% and
%
%\bigskip

%%\section*
%\paragraph{\it Acknowledgments.}
% {\bf(do this as a footnote on page 1?)}

\section{Recursive monotonicity estimate and proof of Theorem \ref{thm:msb}} \label{sec:msb-pf}
%{Summary of the proof
\subsection{Propagation observables}\label{sec:prop-est}
By the linearity of the vNL equation \eqref{vNLeq}, it suffices to consider the evolution $ \rho_t :=e^{Lt}\lam$, with $\lam$ satisfying $\chi_{A_b}\lam=0$. %(\tilde{\chi}^{-}_b \rho_{\rm nst} \tilde{\chi}^-_b)$ for some $\rho_{\rm nst} \ge 0$, $\rho_{\rm nst}\in\mathcal{D}$. 

%\bigskip
 %\smallskip
%{\it Approach.}  
Our goal is to estimate $\Tr(\chi_{{\eta}} e^{L t}\lam)$.  
 To this end, we use the method of propagation observables.

 Let $\rho_t=e^{L t}\lam$ be the solution to 
the vNL equation \eqref{vNLeq} and denote the average of $A$ in the state $\rho_t$ 
 by  \[\lan A\ran_t :=\Tr(A\rho_t).\]  
 We consider a time-dependent, non-negative 
operator-family ({\it propagation observable}) $\Phi_t$ and try to obtain {\it propagation estimates} 
of the form $0\le \Tr(\Phi_t\rho_t)\ls t^{-n}$.

By the definition $\Tr(A L\rho)=\Tr((L'A)\rho)$, we have the relation  
\begin{align}\label{dt-Heis}
&{d\over{dt}}\left<\Phi_t\right>_t =\lan D \Phi_t\ran_t; \qquad D \Phi_t=L' \Phi_t 
 +\partial_t\Phi_t,
\end{align}
for all $t$, provided $\Phi_t\in\mathcal{D}(L')$ and $\partial_t \Phi_t\in \mathcal B(\mathcal H)$. We call $D$ the {\it Heisenberg derivative}.  We would like to show that  $ D \Phi_t \le 0$, modulo fast time-decaying and recursive terms  %self-similar  of a specific structure, 
(see \eqref{DPhi-est1} below),  in which case the relation 
\begin{align} \label{eq-basic}  
\lan \Phi_t\ran_t-\int_0^t \lan D \Phi_r\ran_r dr= \lan \Phi_0\ran_0,
\end{align}
which follows from 
the equation $\lan \Phi_t\ran_t= \lan \Phi_0\ran_0+\int_0^t \frac{d}{dr} \left<\Phi_r\right>_r dr$ and \eqref{dt-Heis}, %If $ D \Phi_t \le 0$, modulo fast time-decaying terms,  then this relation
  gives estimates on the positive terms $ \lan \Phi_t\ran_t$ and $-\int_0^t \lan D \Phi_r\ran_r dr$.
We call \eqref{eq-basic} the basic equality.

% Namely, we consider propagation observables of the form 
%\begin{align}\label{Phi}\Phi_{ts}= f(x_{ts}),\end{align}  %where  $$ x_{ts}=s^{-1}(\x -a-c t),  $$
% with $s\ge t$, and 
% $f$ is a non-negative, non-decreasing function supported in $(0, \infty)$. 

\subsection{Function spaces}
%{\it Definitions and notation}. 
We fix $c'$ such that $c>c'>\kappa$ 
and %$\del>0$ with  $\del\le c - v$ and
 let $\cF$ be the set of functions  $0\le f\in C^\infty(\R)$, supported in $\R^+$ and
satisfying $f(\mu)=1$ for $\mu\ge c-c'$,  and $f^\prime\ge 0$, with $\mathrm{supp}(f')\subset(0,c-c')$ and $\sqrt{f'}\in C^\infty$. %Here and in the rest of the section, $n$ is the same as in the main result \eqref{max-vel-est}.
 
Moreover, besides the notation of Theorem \ref{thm:msb}, we will use the following notation %(for any $\xi\ge 0, \Tr\ \xi=1$): 
\begin{align}\label{fts-not}
%\chi^-_b:=\chi_{A_{b}^-},\quad &\rho_t :=e^{-iLt}\rho_0,\ \rho_0:=\chi^-_b\xi \chi^-_b \quad \text{and} \quad x_{ts} :=s^{-1}(\x -a -v t),\\
&f_{ts}:=f(x_{ts}), \quad   f_{ts}'=(f')_{ts}, \quad \text{ where }\ x_{ts}:=s^{-1}(\x -a-c' t)%,  \qquad x_{ts} :=s^{-1}(\x -a -c t),  
\end{align}
%We will  use the convention that 
%$A\dot\le B$ and $A\dot\ls B$ mean that for any $n>0$, there is $C>0$ s.t. 
%$A \le B+ C s^{-n}$ and $A \ls B+ C s^{-n}$, respectively. 
The factor $s^{-1}$ is introduced to control multiple commutators and commutator products. 
It can be thought of as an adiabatic or semi-classical parameter. %{\bf(change $s^{-1}$ to $\ka$?)} 

\subsection{Recursive monotonicity %\sout{propagation} 
estimate and proof of the main result}

The notation $O(s^{-m})$ denotes an operator $R\in \mathcal{B}(\mathcal{H})$ such that $\|R\|\leq C s^{-m}$ uniformly in $0\leq t\leq s$.
The following is the key estimate underlying the proof of Theorem \ref{thm:msb}:
\begin{prop}[Recursive monotonicity estimate]\label{prop:qm-est1}Assume the hypotheses of Theorem \ref{thm:msb}. % and let  $f\in \cF$, $f_{ts}'=(f')_{ts}$ and $u_{ts}=(f_{ts}')^{1/2}$.
Then, for any $f\in \cF$, there is $C>0$ and $\tilde f\in \cF$ s.t. %an admissible function $h$ s.t. 
\begin{align}\label{DPhi-est1}
&D f_{ts}  %\notag\\  & 
\le (\kappa -c')s^{-1}  f'_{ts}  +   C s^{-2}  \tilde f'_{ts}+O(s^{-n}). %,\qquad h_{ts}:=\tilde u_{ts}^2+v_{ts}^2.
\end{align}
\end{prop}
Proposition \ref{prop:qm-est1} is proved in Section \ref{sec:rme-pf} below. We call \eqref{DPhi-est1} the {\it recursive monotonicity} %{\it main, or quasi-monotonicity} %\sout{propagation} 
estimate.  In the next proposition, we integrate this estimate.
\begin{prop}[Propagation estimate]\label{prop:propag-est1} 
Under the hypotheses of Theorem \ref{thm:msb}, for any $f\in \cF$, there is $\tilde f\in \cF$ and $C>0$ such that for all fixed $t>0$ and all $s\ge t$, % $\supp\tilde f =\supp f$ and 
\begin{equation} \label{propag-est2}
\lan f_{ts}\ran_t+(c'- \kappa)s^{-1}\int_0^t \lan  f_{ts}'\ran_r dr \le\,C s^{-2}\int_0^t \lan  \tilde f_{rs}'\ran_r dr+O(s^{1-n}).
\end{equation}
%{\bf (FP) Have put $\ls$ instead of $\le$}
%for any $n$, uniformly in $t\le s$, $1\le s<\infty$ and $a >b>0$.
\end{prop} 
 
\begin{proof}[Proof of Proposition \ref{prop:propag-est1}]Taking the trace of  \eqref{DPhi-est1} with respect to the density operator 
$\rho_t=e^{tL}\lam$,  we obtain % (\tilde{\chi}^-_b\rho_{\rm nst} \tilde{\chi}^-_b)$,  we obtain
  \[
\lan Df_{ts}\ran_t\, \le \, (\kappa -c')s^{-1} \lan f_{ts}' \ran_t + C s^{-2} \lan \tilde f'_{ts} \ran_t+O(s^{-n}).
\]
Integrating this over time and using Eqs. \eqref{dt-Heis} and \eqref{eq-basic} we find
$$
\lan f_{ts}\ran_t+(c'-\kappa)s^{-1}\int_0^t \lan  f_{ts}'\ran_r dr\notag \le\, 
\lan f_{0s}\ran_0+C s^{-2}\int_0^t \lan  \tilde f'_{rs}\ran_r dr+O(ts^{-n}).
$$

Finally, we claim that, for any $ f\in \mathcal F$,
\begin{align}\label{local-est4}\lan  f_{0 s}\ran_0=\Tr(f_{0 s}\lam) =0. %,\qquad n\geq 1,\, f\in \mathcal F.
\end{align} 
 %and the definition $\Phi_{ts}:= f(x_{ts})$ give  
To see this, recall that for any $f\in \cF$, we have  $\supp f\subset \R^+$
and therefore $\supp f_{0 s}\subset \{\x\ge a+\delta s\}$. Since $\chi_{A_b}\lam =0$ and $b<a$, we have that %$\supp {\color{blue}\tilde{\chi}^-_b}\subset \{\x\le b\}$ and $b<a$, the functions $f_{0 s}$ and ${\color{blue}\tilde{\chi}^-_b}$ have disjoint supports. 
%Hence, %we have by \eqref{comm-exp-x} of Lemma \ref{lem:commut-exp}  (and the fact that $f^{(k)}_{0s}\chi^-_b=0$)  that, for any $n$, 
\begin{align}\label{local-est3}
f_{0 s} \lam %{\color{blue}\tilde{\chi}^-_b} %=[f(x_{0 s}),g(H)]\chi^-_=O(s^{-n}).
=0. 
\end{align}
This proves \eqref{local-est4} and therefore \eqref{propag-est2}.
%Since $k < v$ and since $h=\tilde u^2$  is an admissible function and therefore $h\gs j'$ for some $j\in \cF$, \eqref{propag-est2} implies that
%\begin{align}\label{propag-est3} 
%\int_0^t \lan  f'(x_{rs})\ran_r dr\, 
%  \dot\ls \, %(v-k)^{-1}
%  s \lan f_{0 s}\ran_0
%  + s^{-1}\int_0^t \lan  j'_{rs}\ran_r dr. 
%\end{align}
%(The constant entering the relation $\ls$ here is bounded by a power of $(v-k)^{-1}$.)
%This, together with the boundedness of $j'$,  
%gives estimate \eqref{propag-est1} with $n=0$.
%Now, assuming \eqref{propag-est1} holds for some $n=n'\ge 0$, 
%  and using this (with $f=j$) for the integral on the r.h.s. of \eqref{propag-est3}, we see that  
%\eqref{propag-est1} holds for $n=n'+1$.  
\end{proof} 

We now show that Proposition \ref{prop:propag-est1} yields the main result via an iteration argument.

\begin{proof}[Proof of Theorem \ref{thm:msb} (assuming Proposition \ref{prop:qm-est1})]
We consider an arbitrary $f\in \mathcal F$. Dropping the second term in \eqref{propag-est2} yields
\begin{equation}\label{eq:firststep}
\lan f_{ts}\ran_t 
\le\,C s^{-2}\int_0^t \lan  \tilde f_{rs}'\ran_r dr+O(s^{1-n}).
\end{equation}

Now we apply Proposition \ref{prop:propag-est1} to the function $\tilde f$. It yields another function $\hat f\in \mathcal F$ such that, after we drop the first term,
 \[
(c'-\kappa)s^{-1}\int_0^t \lan  \tilde f_{rs}'\ran_r dr 
\le\,C s^{-2}\int_0^t \lan  \hat f_{rs}'\ran_r dr+O(s^{1-n}).
\]
Since $c'-\kappa>0$ by assumption, it can be absorbed into the constant $C$. We iterate this procedure $n-1$ times and bound the final integral by the a priori bound $t\leq s$ using that the derivative of any function in $\mathcal F$ is uniformly bounded. This gives
 \[
s^{-1}\int_0^t \lan  \tilde f_{rs}'\ran_r dr 
\le\,C s^{-n}\int_0^t \lan  \hat f_{rs}'\ran_r dr+O(s^{1-n})\leq Cs^{1-n}.
\]
We apply this estimate to \eqref{eq:firststep} and find
\begin{align} \label{propag-est4} 
\lan f_{ts}\ran_t \le\,C s^{-2}\int_0^t \lan  \tilde f_{rs}'\ran_r dr+O(s^{1-n})\leq  C s^{1-n}.
\end{align}

It remains to choose $s$ appropriately. For any $f\in \cF $, we have $f(\mu)=1$ for $\mu\ge c-c'$, and therefore $f(x_{t s})=1$ on $\{\x\ge a+c't +(c-c') s\}$.
Recall 
the notation $A_\eta :=\{x\in \R^d:  \x \ge  \eta\}$. By our assumption, $\eta\ge a+ct$. 
 We  set 
  \[
 s=(\eta-a)/c\ge t.
 \]
  This gives $\eta=a +c s\ge a +c t$ and therefore   
\begin{equation*}
A_{\eta} \subset  \{\x\ge a+c't +(c-c') s\}\subset\{f(x_{t s})=1\}.
\end{equation*}  
Using this together with estimate \eqref{propag-est4}   
  and the definitions  $s=(\eta-a)/c$, 
we obtain  that $\lan \chi_{\eta}\ran_t \le C\eta^{1-n}$, which implies \eqref{max-vel-est}.   %we arrive at \eqref{max-vel-est}. %{max-vel-est-info}.
   \end{proof}

\section{Proof of the recursive monotonicity
 estimate} \label{sec:rme-pf}

\begin{proof}[Proof of Proposition \ref{prop:qm-est1}]
In what follows, we often denote $\sum\nolimits_{j\geq 1}\equiv \sum\nolimits_{j}$. We use  the time-dependent observable 
\begin{align}\label{propag-obs1}
\Phi_{ts} & :=f_{ts}\equiv f(x_{ts}), \quad f \in \cF, %\ f'\gs \tilde f, 
\end{align}
with $0\le t\le s$.  First, we observe that the operator $L' $  is given explicitly by
\begin{align}\label{L'}
& L' =L'_0+G',\  \qquad L'_0A %=&i[H,A]+\sum_{j\geq 1} (W_j^* A W_j-\frac{1}{2} \{W_j^* W_j,A\})\\
=i [H,A],\\
\label{G'}& G'A:=\frac{1}{2}\sum_{j\geq 1}(W_j^* [A, W_j]+[W_j^*, A] W_j),
\end{align}
with domain 
 \begin{align*}
 \mathcal{D}(L')\equiv&\,\mathcal{D}(L'_0)\equiv\big\{A\in \mathcal{B}(\mathcal{H})\, | \, A\mathcal{D}(H)\subset\mathcal{D}(H) \text{ and }\\
&HA-A H \text{ defined on }\mathcal{D}(H) \text{ extends to an element of } \mathcal{B}(\mathcal{H}) \big \}.
 \end{align*}
It follows from Assumption \eqref{Hassmpt:main} %--\eqref{domainassmpt}
  and Lemma \ref{lem:commut-exp} that for all $0\le t \le s$, $\Phi_{ts}\mathcal{D}(H)\subset\mathcal{D}(H)$ and that $[\Phi_{ts},H]$ defined on $\mathcal{D}(H)$
 extends to a bounded operator. Therefore $\Phi_{ts}\in\mathcal{D}(L')$. %Likewise, since \sbb{$\rho_0 = \tilde{\chi}^{-}_b \rho_{\rm nst} \tilde{\chi}^-_b$ with $ \rho_{\rm nst}\in\mathcal{D}(L)$ and $\tilde{\chi}^-_b$} is smooth and bounded, we deduce that $\rho_0\in\mathcal{D}(L)$.
  Hence, 
 in order to estimate $\left<f_{ts}\right>_t=\Tr\,(f_{ts}\rho_t)$,  
we can apply \eqref{dt-Heis} and  the basic equality \eqref{eq-basic}. 
We start by computing $D\Phi_{ts}$. First, we have
\begin{align} \label{dt-Phi}
{\partial\over{\partial t}}f_{ts}=-s^{-1}c' \,f^\prime_{ts}.
\end{align}
 %\sim {1\over 2}(p\cdot \n\x+\n\x\cdot p)$, with $p:=-i\n$.
The more interesting term is the first term in the definition of $D$ in  \eqref{dt-Heis}:
 \[
 L'f_{ts}
 =i [H,f_{ts}]+\frac{1}{2}\sum_{j\geq 1}(W_j^* [f_{ts}, W_j]+[W_j^*, f_{ts}] W_j).
 \]
  The terms on the r.h.s. are controlled via the following two key lemmas. First, it convenient to introduce the following definition. 
 We say a function $h$ is {\it admissible} if it is smooth, non-negative with $\supp h\subset (0, c-c')$ and $\sqrt{h}\in C^\infty$. 
Note that if $h$ is admissible, then %$h= f'$, with f/\int_{-\infty}^\infty h(s)ds\in \cF
\[h= f',\ \text{ with }\ f/f(\infty)\in \cF,\ \text{ where }\  f(\mu)=\int_{-\infty}^\mu h(s)ds.\]   
 
 \begin{lemma}[Estimate of Hamiltonian contribution]\label{lemma:Hestimate} Under the Hypotheses of Proposition~\ref{prop:propag-est1}, 
let $f_{ts}'=(f')_{ts}$ and $u_{ts}=(f_{ts}')^{1/2}$. Then, we have
 \begin{align} \label{H-Phi-comm}
i[H, f_{ts}] %=&{i\over 2}[p^2,\Phi_s(t)] =\frac12 s^{-1}(\g f^\prime(x_{ts})+f^\prime(x_{ts})\g) \notag \\ 
&= s^{-1} u_{ts} i[H, \x] u_{ts}  + \mathrm{Rem}_H
 \end{align} 
 where the remainder satisfies the operator inequality
\begin{align}\label{H-Phi-rem-est} \mathrm{Rem}_H\leq  C  s^{-2} \tilde u_{ts}^2 +O(s^{-n}),
\end{align}
for a suitable admissible function $\tilde u^2$.
 \end{lemma}
 
%that $i[H, \x]=\g$ and $[[\g,u],u]=0$, we find

%This equation, together with Eq. \eqref{dt-Phi}, Condition (b) and Eq. \eqref{L'-spec}, yields:
%\begin{align} \label{DPhi-expr-x}
%Df_{ts}= s^{-1}\,u_{ts}\,(\g-v)\,u_{ts}+\sum_{j\geq 1}
% [[W_j, f_{ts}], W_j].\end{align}
%Now,  with $\kappa$ defined in \eqref{kappa}, we 
%\begin{align}\label{A-est} %\notag 
% u_{ts}\,\g\,u_{ts} \le k &  u_{ts}^2.\end{align} 

The main novelty for the von Neumann-Lindblad equation is the following estimate on the interaction with the environment.

\begin{lemma}[Estimate on the environment contribution]\label{lemma:Westimate}
Under the Hypotheses of Proposition~\ref{prop:propag-est1} and with the definition \eqref{G'}, 
\begin{align}\label{doubl-commW-est}
&%\sum_{j\geq 1} \left(W_j^* [f(x_{ts}), W_j]+[W_j^*, f(x_{ts})] W_j\right)\notag\\ &=s^{-1} u_{ts} \left(\sum_{j\geq 1}\left(W_j^* [\x, W_j]+[W_j^*, \x] W_j\right)\right)u_{ts}
G' f_{ts}=s^{-1} u_{ts} (G' \x) u_{ts}+\mathrm{Rem}_W
\end{align}
 where the remainder satisfies the operator inequality
\begin{align}\label{doubl-commWrem-est} \mathrm{Rem}_W\leq  C  s^{-2} v_{ts}^2 +O(s^{-n}),
\end{align}
for a suitable admissible function $v^2$.%\cms{conflict with the notation for velocity}
\end{lemma}

These lemmas will be proved in Subsections \ref{sect:Hestimate} and \ref{sect:Westimate} below by using the commutator expansion in Lemma \ref{lem:commut-exp} several times.  This lemma is applicable due to  Assumptions \eqref{Hassmpt:main'} and \eqref{Hassmpt:main}.
\DETAILS{ the following 
relation 
\begin{equation}\label{domainassmpt}
\langle x\rangle^{-1}\mathcal{D}(H)\subset\mathcal{D}(H),
\end{equation}
which follows from Assumption \eqref{Hassmpt:main} for $k=1$, the fact that $\mathcal{D}(H)=\Ran (H+i)^{-1}$ and the relation
 \begin{align}\label{resolv-rel}
\langle x\rangle^{-1}(H+i)^{-1}= &(H+i)^{-1}\langle x\rangle^{-1}\notag\\
&+ (H+i)^{-1}\langle x\rangle^{-1} [\langle x\rangle, H] \langle x\rangle^{-1}(H+i)^{-1}.
\end{align}}
%imply that the operator $\langle x\rangle^{-1}$ leaves the domain of $H$ invariant \eqref{domainassmpt}  

Combining \eqref{dt-Phi}, \eqref{H-Phi-comm}, \eqref{H-Phi-rem-est}, \eqref{doubl-commW-est}, and \eqref{doubl-commWrem-est} %{H-Phi-comm} and \eqref{doubl-commW-est},
 and recalling the definitions $u_{ts}^2=f'_{ts}$ and of $\kappa$ in \eqref{kappa}, we obtain \eqref{DPhi-est1}.\end{proof}

\subsection{Proof of Lemma \ref{lemma:Hestimate}}\label{sect:Hestimate}
Thanks to Assumptions \eqref{Hassmpt:main'} and \eqref{Hassmpt:main}, % and relation \eqref{domainassmpt}, 
 we can use Lemma \ref{lem:commut-exp}, more precisely \eqref{Hcomm-exp} and its adjoint, to obtain the commutator expansion
 \begin{align}\label{Hcomm-exp'} [H,  f(x_{ts})]&=  \sum_{1\leq k < n}{s^{- k}\over{k!}} f^{(k)}(x_{ts})B_k+O(s^{-n}\|B_n\|)\end{align}
 %uniformly in $a > 0$, 
  where $B_k= \,{\mathrm{ad}_{\x}^k H}$.

In order to further use estimates on $B_k$ from Assumption \eqref{Hassmpt:main}, we need to symmetrize the appearance of the derivative. We set $u_1=\sqrt{f'}\geq 0$ which satisfies $u_1\in C^\infty(\R_+)$ since $f\in\mathcal F$. 
Furthermore, 
for $k\geqq 2$, %we use that $f^{(k)}\prec u_k$ to write $f^{(k)}=f^{(k)}u_k^2$ with $u_k\in C_c^\infty(\R_+)$.
we let $u_k\in C_c^\infty(\R_+)$ be s.t.  $u_k=1$ on $ \mathrm{supp}\, f^{(k)}$.

We factor $f^\prime= u_1^2$ and write $f^{(k)}= u^2_k g_k$ with $g_k=f^{(k)}$ for $k\geq 2$ and $g_1=1$. Then we write 
$$
f^{(k)}(x_{ts})B_k = u_k(x_{ts}) g_k(x_{ts}) B_k u_k(x_{ts}) + u_k(x_{ts}) g_k(x_{ts})[u_k(x_{ts}),B_k] , \quad 1\leq k\leq n.
$$
We can again expand the commutator via Lemma \ref{lem:commut-exp},
\begin{align}\label{Bcomm-exp}
[  u_k(x_{ts}),B_k]= - \sum_{m=1}^{n-k-1} {(-1)^m} {s^{- m}\over{m!}} u_k^{(m)}(x_{ts})B_{k+ m}+O(\|B_n\|s^{-n+k}).
\end{align}

Iterating this symmetrization procedure, we find %, {\bf using that  $[[\g,u],u]$ is bounded, we find}
\begin{align} \label{H-Phi-comm'}
i[H, \Phi_{ts}]\ %=& \frac12 s^{-1}(\g f^\prime(R_{ts})+f^\prime(R_{ts})\g)\notag \\ 
& = s^{-1}\,u_1(x_{ts})\, [iH, \x]\,u_1(x_{ts})+\mathrm{Rem}_H
\end{align}
where 
\begin{align}
\mathrm{Rem}_H= \sum_{k=2}^{n-1} s^{- k}v_k(x_{ts}) D_k v_k(x_{ts})+O(s^{-n}\|B_n\|),
\end{align}
 with $v_k\in C^\infty_c((0,c-c'))$, $v_k=1$ on $\mathrm{supp}(f')$ and $D_k$ bounded operators satisfying 
 \begin{align}\label{eq:Dkbound}
  \|D_k\|\le C_k\|B_k\|,\qquad 2\leq k\leq n-1
  \end{align} To obtain an operator bound from this norm bound, we rewrite \eqref{H-Phi-comm'} with a manifestly self-adjoint remainder term,
 \begin{align} 
i[H, \Phi_{ts}]
=&\frac{1}{2}\big(i[H, \Phi_{ts}]+(i[H, \Phi_{ts}])^*\big)\notag \\
%=& \frac12 s^{-1}(\g f^\prime(R_{ts})+f^\prime(R_{ts})\g)\notag \\ 
& = s^{-1}\,u_1(x_{ts})\, [iH, \x]\,u_1(x_{ts})+\frac{1}{2}\big(\mathrm{Rem}_H+\mathrm{Rem}_H^*\big) \label{H-Phi-comm''}
\end{align}
where 
\begin{align}
\frac{\mathrm{Rem}_H+\mathrm{Rem}_H^*}{2}= \sum_{k=2}^{n-1} s^{- k}v_k(x_{ts}) \frac{D_k+D_k^*}{2} v_k(x_{ts})+O(s^{-n}\|B_n\|).
\end{align}
Thanks to self-adjointness and \eqref{eq:Dkbound}, we have the operator inequality 
$$
\frac{D_k+D_k^*}{2} 
\leq \|D_k+D_k^*\|
\leq 2 \|D_k||
\leq 2C_k  \|B_{k}\|,\qquad 2\leq k\leq n-1,
$$
and each $\|B_k\|$ is finite by Assumption \eqref{Hassmpt:main}. 
This implies
\begin{align}
\frac{\mathrm{Rem}_H+\mathrm{Rem}_H^*}{2}\leq C_n s^{-2} \max_{0\leq k\leq n} \|B_k\|\sum_{2\leq k\leq n} s^{-k+2} v_k(x_{ts})^2.
\end{align}
Since $\sum_{2\leq k\leq n} s^{-k+2} v_k^2$ is bounded by an $s$-independent admissible function, this proves Lemma \ref{lemma:Hestimate}.
 \qed

\subsection{Proof of Lemma \ref{lemma:Westimate}}\label{sect:Westimate}

\begin{proof}[Proof of Lemma \ref{lemma:Westimate}]
Fix $j\geq 1$ and recall the definition \eqref{G'}. We can restrict our attention to a single term $W_j^* [f(x_{ts}), W_j]$ and take adjoints and a sum over $j$ at the end to derive the lemma. Using Lemma \ref{lem:commut-exp}, we obtain the commutator expansion
\begin{align}\label{eq:initialWexpansion}
W_j^* [f(x_{ts}), W_j]=W_j^*\sum_{k=1}^{n-1}{s^{- k}\over{k!}}f^{(k)}(x_{ts})A^k_j +W_j^*O(s^{-n}\|A^{n}_j\|).
 %uniformly in $a > 0$, 
\end{align}
  where $A^k_j= \,{\mathrm{ad}_{\x}^k W_j}$. Notice that the last error term is summable in $j$ by Assumption \eqref{Wassmpt:main} and the Cauchy-Schwarz inequality and yields $O(s^{-n})$, so it can be ignored in the following.

  We consider the first term on the right-hand side of \eqref{eq:initialWexpansion} and symmetrize the expression to the right of $W_j^*$. To this end, we write $f^{(k)}= u^2_k g_k$ with $u_k$ and $g_k$ defined as in the proof of Lemma \ref{lemma:Hestimate}. Then we write, for $k\geq 1$,
\begin{multline}\label{eq:Wsymm}
W_j^* f^{(k)}(x_{ts})A_j^k =  u_k(x_{ts}) W_j^* g_k(x_{ts}) A_j^k u_k(x_{ts})  \\ 
+u_k(x_{ts}) W_j^* g_k(x_{ts}) [ u_k(x_{ts}) ,A_j^k]
+ [W_j^*, u_k(x_{ts})  ]g_k(x_{ts}) u_k(x_{ts}) A_j^k. %, \qquad k\geq 1.
\end{multline}
We can again expand the first commutator via Lemma \ref{lem:commut-exp},
\begin{align}\label{Bcomm-exp'}
[  u_k(x_{ts}),A_j^k]=-  \sum_{m=1}^{n-k-1}{{(-1)^m}s^{- m}\over{m!}} u_k^{(m)}(x_{ts})A_j^{k+m}+O(s^{-n+k}\|A_j^{n}\|).
\end{align}
For the second commutator in \eqref{eq:Wsymm}, we note that $W_j^*=(A_j^0)^*$ and use the adjoint version of Lemma \ref{lem:commut-exp},
\begin{align}\label{Bcomm-exp''}
[(A_j^k)^*,  u_k(x_{ts}) ]&=[u_k(x_{ts}) ,A_j^k]^*\notag\\
&=  \sum_{m=1}^{n-k-1}{s^{- m}\over{m!}} (A_j^{k+m})^*u_k^{(m)}(x_{ts})+O(s^{-n+k}\|A_j^{n}\|),
\end{align}

Iterating this symmetrization procedure, we find
\begin{align}
  W_j^* [f(x_{ts}), W_j]
  = s^{-1}   u(x_{ts})W_j^* [\x,W_j] u(x_{ts})+\mathrm{Rem}_{W,j}
\end{align}
with
 \[
\mathrm{Rem}_{W,j}= \sum_{k=2}^{n-1} s^{- k}v_k(x_{ts}) D_j^k v_k(x_{ts})+O(s^{-n}\|D_j^{n}\|).
\]
Here $v_k\in C^\infty_c((0,c-c'))$ and $v_k=1$ on $\mathrm{supp}\, f'$, $v_k$ are independent of $j$, and $D^k_j$ are bounded operators satisfying the norm bound
\begin{align}\label{eq:Dkjbound}
 \|D^k_j\|\le C_ka_j^k,
\end{align}
where we introduced the shorthand
\[
a_j^k:=\max_{\substack{0\leq \ell,m\leq k:\\ \ell+m=k}} \|A^{\ell}_j\| \|A^{m}_j\|.
\]

 We take the adjoint relation to find
\begin{align}
 W_j^* [f(x_{ts}), W_j]+&[W_j^*, f(x_{ts})] W_j
= s^{-1}   u(x_{ts})\big( W_j^* [\x,W_j]\notag\\
&\quad+[W_j^*,\x]W_j \big) u(x_{ts})  +\mathrm{Rem}_{W,j}+(\mathrm{Rem}_{W,j})^*.
\end{align}
Now we take the sum over $j\geq 1$ (whose convergence is justified a posteriori) and recall the notation $
u_{ts}:=u(x_{ts})$ (see \eqref{fts-not}) to obtain \eqref{doubl-commW-est} 
%\begin{align} &\sum_{j\geq 1}\left( W_j^* [f(x_{ts}), W_j]+[W_j^*, f(x_{ts})] W_j \right)\\ =&s^{-1}   u(x_{ts})\sum_{j\geq 1}\left( W_j^* [\x,W_j]+[W_j^*,\x]W_j \right) u(x_{ts}) +\mathrm{Rem}_{W} \end{align}
with the remainder 
$$
\mathrm{Rem}_{W}=\sum_{j\geq1}\left(\sum_{k=2}^{n-1} s^{- k}v_k(x_{ts}) ( D_j^k +(D_j^k)^*)v_k(x_{ts})+O(s^{-n}\|D_j^{n}\|)\right).
$$
From self-adjointness and the norm bound \eqref{eq:Dkjbound}, we conclude the operator inequality $D_j^k +(D_j^k)^*\leq 2C_k a_j^k$ and hence
$$
\mathrm{Rem}_{W}\leq 
C_ns^{-2} \sum_{k=2}^{n-1} \left( \sum_{j\geq1}  a_j^k\right)  s^{- k+2} v_k(x_{ts}) v_k(x_{ts})+O(s^{-n})\sum_{j\geq 1}a_j^n.
$$
Assumption \eqref{Wassmpt:main} implies that $\sum_{j\geq1} a_j^k<\infty$ for each $k\geq 2$. To complete the proof, it remains to note that $ \sum_{k=2}^{n-1} s^{- k+2} v_k(x_{ts}) v_k(x_{ts})$ is bounded from above by an $s$-independent admissible function. 
\end{proof}

%\newpage

%%%%%%%%%%%%%%%%%%%%%%%%%%%%%%%%%%%%%%%%%%%%%%%%%%%%%%%%%%%%%%%%%

\appendix

\section{Commutator expansions}\label{sec:commut}
In this appendix, we present commutator expansions and estimates, first derived in \cite{SigSof} 
and then improved in \cite{GoJe,HunSig1, HunSigSof, Skib}. 
We follow \cite{HunSig1} and refer to this paper for %the definitions, 
details and references.
Here, we mention only that, by the Helffer-Sj\"ostrand formula, 
a function $f$ of a self-adjoint operator $A$ can be written as
\begin{align} \label{fA-repr}
&f(A)=\int d\widetilde f(z)(z-A)^{-1},
\end{align}
where $\widetilde f(z)$  is an almost analytic extension of $f$ to $\C$ supported in a complex neighbourhood of $\supp f$ \cite{HelffSj}. 
For $f\in C^{n+2}(\R)$, we can choose $\widetilde f$ satisfying the estimates (see (B.8) of \cite{HunSig1}, see also \cite{DerGer,Dav,IvrSig}):
\begin{align} \label{tildef-est}
&\int |d\widetilde f(z)||\im(z)|^{-p-1}\ls \sum_{k=0}^{n+2}\|f^{(k)}\|_{k-p-1},\end{align}
where $\|f \|_{m}:=\int \x^m |f(x)| dx$ and any integer $0\leq p \leq n$.

The essential commutator expansions and remainder estimates are incorporated in the following lemma:

\begin{lemma}\label{lem:commut-exp} 
Let $f\in C^\infty(\R)$ be bounded, 
  with $\sum_{k=0}^{n+2}\|f^{(k)}\|_{k-2}<\infty$, for some $n\ge 1$.
Let  $x_s=s^{-1}(\langle x\rangle
-a)$ for $a >0 $ and
$1\le s<\infty$. Let $A$ be an operator such that $\langle x \rangle^{-1}\mathcal{D}(A)\subset\mathcal{D}(A)$.
 %Suppose that $A \in \mathrm{C}^n$. %$H$ satisfies \eqref{V-cond}.
Define
 \[
B_k= \,{\mathrm{ad}_{\x}^k A}
\] 
and assume that $\|B_k\|<\infty$ for all $1\leq k\leq n$.
Then, for any $n\ge 1$, %, under the hypothesis of Theorem \ref{thm:min-vel-est}:
 \begin{align}\label{Hcomm-exp} [A, f(x_s)]&=  \sum_{1\le k\le n-1}(-1)^{k-1}{s^{- k}\over{k!}}B_kf^{(k)}(x_s)  +O(s^{-n}\|B_n\|), \end{align}
 uniformly in $a\in \R$. %$H^j B_k, j=0, 1, k =1, \dots, n-1,$ are bounded operators and $\|H^j O(s^{-n})\|\ls s^{-n}, j=0, 1$. 
% For $n=1$, the sum on the r.h.s. should be omitted. 
\end{lemma}

\begin{proof}
Using \eqref{fA-repr}, we have
\begin{align*}
[A, f(x_s)]=\int d\widetilde f(z) \big[ A , (z-x_s)^{-1} \big ],
\end{align*}
in the sense of quadratic forms on $\mathcal{D}(A)$. The hypothesis $\langle x\rangle^{-1} \mathcal{D}(A)\subset\mathcal{D}(A)$ shows that  $(z-x_s)^{-1}=\x^{-1}(z \x^{-1}-x_s\x^{-1})^{-1}$ maps $\mathcal{D}(A)$ into itself for $z$ with large $|\im z|$ and therefore for all $z$ with $\im z\neq 0$. Hence, since $[A,\langle x\rangle]=\mathrm{ad}^1_{\langle x\rangle}(A)$ is bounded, the formula
\begin{equation*}
s\big[ A , (z-x_s)^{-1} \big ]=(z-x_s)^{-1}[A,\langle x\rangle](z-x_s)^{-1}
\end{equation*}
holds in the sense of quadratic forms on $\mathcal{D}(A)$ ($\Im\lan Au, B^{-1} u\ran=\Im\lan u, A B^{-1} u\ran=\Im\lan  B B^{-1} u, A B^{-1} u\ran$).
Since $[A,\langle x\rangle]=\mathrm{ad}^1_{\langle x\rangle}(A)$ is bounded, we can proceed as in
 (B.14)-(B.15) of \cite{HunSig1}, commuting successively the commutators $ad^1_{\langle x\rangle}(A)$ to the left. This yields  
 \begin{align*}
[A, f(x_s)]&=\sum_{1\le k\le n-1}(-1)^{k-1}{s^{- k}\over{k!}} B_kf^{(k)}(x_s)+s^{-n}\Re(s);\\
%\label{Bk-bnd}B_k&= \,{ad_{\x}^kg(H)};\\ %\hbox{ \it (bounded for every } k);\\
 \Re(s)&=\int d\widetilde f(z)(z-x_s)^{-1}B_n(z-x_s)^{-n}.
     \end{align*}    
%where for $n=1$, the sum on the r.h.s. is supposed to be omitted. 

%Now we show that the operators %$H^j B_k, j=0, 1, k =1, \dots, n,$ and  
%$Re(s)$ are bounded.
  Since the operator $B_n$ is bounded, we have
 \begin{align*} \| \Re(s)\|&\le  \| B_n\|  \int |d\widetilde f(z)| |z-x_s|^{-n-1}\\
 &\le  \| B_n\|  \int |d\widetilde f(z)| |\im z|^{-n-1}\ls \| B_n\| \sum_{k=0}^{n+2}\|f^{(k)}\|_{k-n-1}.\end{align*}
 This concludes the proof.
\end{proof}

\end{document}